\documentclass[10pt,onecolumn,draftcls,journal]{IEEEtran}
\usepackage[ansinew]{inputenc}
\usepackage[pdftex]{graphicx}
\usepackage[T1]{fontenc}
\usepackage{color,amsmath,enumerate,amssymb,hhline,verbatim}
\usepackage{url}

\begin{document}

\title{Codes between MBR and MSR Points with Exact Repair Property \thanks{Part of this paper was presented at 2013 IEEE Information Theory Workshop, Seville, Spain \cite{exacternvall}.}}

\author{Toni Ernvall \thanks{T. Ernvall is with the Turku Centre for Computer Science, Turku, Finland and with the Department of Mathematics and Statistics, FI-20014, University of Turku, Finland (e-mail:tmernv@utu.fi).}}

\maketitle

\newtheorem{definition}{Definition}[section]
\newtheorem{thm}{Theorem}[section]
\newtheorem{proposition}[thm]{Proposition}
\newtheorem{lemma}[thm]{Lemma}
\newtheorem{corollary}[thm]{Corollary}
\newtheorem{exam}{Example}[section]
\newtheorem{conj}{Conjecture}
\newtheorem{remark}{Remark}[section]

\newcommand{\La}{\mathbf{L}}
\newcommand{\h}{{\mathbf h}}
\newcommand{\Z}{{\mathbf Z}}
\newcommand{\R}{{\mathbf R}}
\newcommand{\C}{{\mathbf C}}
\newcommand{\D}{{\mathcal D}}
\newcommand{\F}{{\mathbf F}}
\newcommand{\HH}{{\mathbf H}}
\newcommand{\OO}{{\mathcal O}}
\newcommand{\G}{{\mathcal G}}
\newcommand{\A}{{\mathcal A}}
\newcommand{\B}{{\mathcal B}}
\newcommand{\I}{{\mathcal I}}
\newcommand{\E}{{\mathcal E}}
\newcommand{\PP}{{\mathcal P}}
\newcommand{\Q}{{\mathbf Q}}
\newcommand{\M}{{\mathcal M}}
\newcommand{\separ}{\,\vert\,}
\newcommand{\abs}[1]{\vert #1 \vert}

\begin{abstract}
In this paper distributed storage systems with exact repair are studied. A construction for regenerating codes between the minimum storage regenerating (MSR) and the minimum bandwidth regenerating (MBR) points is given. To the best of author's knowledge, no previous construction of exact-regenerating codes between MBR and MSR points is done except in the work by Tian et al. On contrast to their work, the methods used here are elementary.

In this paper it is shown that in the case that the parameters $n$, $k$, and $d$ are close to each other, the given construction is close to optimal when comparing to the known functional repair capacity. This is done by showing that when the distances of the parameters $n$, $k$, and $d$ are fixed but the actual values approach to infinity, the fraction of the performance of constructed codes with exact repair and the known capacity of codes with functional repair, approaches to one. Also a simple variation of the constructed codes with almost the same performance is given.
\end{abstract}

\section{Introduction}
\subsection{Regenerating Codes}
In a distributed storage system a file is dispersed across $n$ nodes in a network such that given any $k\, (<n)$ of these nodes one can reconstruct the original file. We also want to have such a redundancy in our network that if we lose a node then any $d\, (<n)$ of the remaining nodes can repair the lost node. We assume that each node stores the amount $\alpha$  of information, \emph{e.g.}, $\alpha$ symbols over a finite field, and in the repair process each repairing node transmits the amount $\beta$ to the new replacing node (called a \emph{newcomer}) and hence the total repair bandwidth is $\gamma=d \beta$. We also assume that $k \leq d$.

The repair process can be either functional or exact. By functional repair we mean that the nodes may change over time, \emph{i.e.}, if a node $v_{i}^{\text{old}}$ is lost and in the repair process we get a new node $v_{i}^{\text{new}}$ instead, then we may have $v_{i}^{\text{old}} \neq v_{i}^{\text{new}}$. If only functional repair is assumed then the capacity of the system, denoted by $C_{k,d}(\alpha,\gamma)$, is known. Namely, it was proved in the pioneering work by Dimakis \emph{et al.} \cite{dimakis} that
$$
C_{k,d}(\alpha,\gamma)=\sum_{j=0}^{k-1} \min \left\{ \alpha, \frac{d-j}{d}\gamma  \right\}.
$$

If the size of the stored file is fixed to be $B$ then the above expression for the capacity defines a trade-off between the node size $\alpha$ and the total repair bandwidth $\gamma$. The two extreme points are called the minimum storage regenerating (MSR) point and the minimum bandwidth regenerating (MBR) point. The MSR point is achieved by first minimizing $\alpha$ and then minimizing $\gamma$ to obtain
\begin{equation}\label{MSR}
\left\{
  \begin{array}{l}
\alpha = \frac{B}{k} \\
\gamma = \frac{d B}{k(d-k+1)}.
  \end{array} \right.
\end{equation}
By first minimizing $\gamma$ and then minimizing $\alpha$ leads to the MBR point
\begin{equation}\label{MBR}
\left\{
  \begin{array}{l}
\alpha = \frac{2d B}{k(2d-k+1)} \\
\gamma = \frac{2d B}{k(2d-k+1)}.
  \end{array} \right.
\end{equation}

In this paper we will study codes that have exact repair property. The concepts of exact regeneration and exact repair were introduced independently in \cite{explicitconst}, \cite{reducingrepair}, and \cite{searchingfor}. Exact repair means that the network of nodes does not vary over time, \emph{i.e.}, if a node $v_{i}^{\text{old}}$ is lost and in the repair process we get a new node $v_{i}^{\text{new}}$, then $v_{i}^{\text{old}} = v_{i}^{\text{new}}$. We denote by
$$
C_{n,k,d}^{\text{exact}} (\alpha,\gamma)
$$
 the capacity of codes with exact repair with $n$ nodes each of size $\alpha$, with total repair bandwidth $\gamma$, and for which each set of $k$ nodes can recover the stored file and each set of $d$ nodes can repair a lost node.

We have by definition that
$$
C_{n,k,d}^{\text{exact}} (\alpha,\gamma) \leq C_{k,d} (\alpha,\gamma).
$$

\subsection{Related Work}
It was proved in \cite{kumar}, \cite{MSRequal}, \cite{optimalMDS}, and \cite{yhteensaMDS} that the codes with exact repair achieve the MSR point and in \cite{kumar} that the codes with exact repair achieve the MBR point. The impossibility of constructing codes with exact repair at essentially all interior points on the storage-bandwidth tradeoff curve was shown in \cite{nonachievability}. Other papers studying exact-regenerating codes in MSR point include \emph{e.g.} \cite{permutation}, \cite{hadamard}, \cite{cadambeoptimalMDS}, and \cite{tamozigzag}. Locally repairable codes that achieve repair bandwidth that falls below the time-sharing trade-off of the MSR and MBR points are studied in \cite{localcodes}.

To the best of author's knowledge, no previous construction of exact-regenerating codes between MBR and MSR points is done except in \cite{exactrepairtian}. Our construction is very different to that. We do not use complex combinatorial structures but instead exploit some optimal codes in MSR point. However, we require in our construction that storage symbols can be split into a sufficiently large number of subsymbols.

Tian has shown in \cite{region433} that there exists a non-vanishing gap between the optimal bandwidth-storage tradeoff of the functional-repair regenerating codes and that of the exact-repair regenerating codes by characterizing the rate region of the exact-repair regenerating codes in the case $(n,k,d)=(4,3,3)$.

\subsection{Organization and Contributions}
In Section \ref{construction} we give a construction for codes between MSR and MBR points with exact repair. In Section \ref{inequalities} we derive some inequalities from our construction. Section \ref{example} provides an example showing  that, in the special case of $n=k+1=d+1$, our construction is close to optimal when comparing to the known capacity when only functional repair is required. In Section \ref{analysis} we show that when the distances of the parameters $n$, $k$, and $d$ are fixed but the actual values approach to infinity, the fraction of performance of our codes with exact repair and the known capacity of functional-repair codes approaches to one.

In Section \ref{easyconstructionsection} we give another construction with quite similar performance. The main differences of this construction when compared to the construction of Section \ref{construction} is its easiness as advantage and relaxation of assumption of symmetric repair as its disadvantage.

In Section \ref{similarconstructions} we give yet two other constructions that have some similarities with the construction of Section \ref{construction}. However, the performance of these constructions is relatively bad and the main interest of this section is the comparison of these constructions with the construction of Section \ref{construction}.

To make it easier to compare our constructions we use notions $P_{n,k,d}^{1}(\alpha,\gamma)$, $P_{n,k,d}^{2}(\alpha,\gamma)$, $P_{n,k,d}^{3}(\alpha,\gamma)$, and $P_{n,k,d}^{4}(\alpha,\gamma)$ to denote the performances of constructions of Section \ref{construction}, Section \ref{easyconstructionsection}, Subsection \ref{konstruktionodekopioinnilla}, and Subsection \ref{konstruktiotiedostoilla}, respectively. It is clear that
$$
P_{n,k,d}^{j}(\alpha,\gamma) \leq C_{n,k,d}^{\text{exact}} (\alpha,\gamma)
$$
for $j=1,2,3,4$.

\section{Main Construction}\label{construction}
Assume we have a storage system $DSS_1$ with exact repair for parameters $$(n,k,d)$$ with a node size $\alpha$ and the total repair bandwidth $\gamma=d\beta$. In this section we propose a construction that gives a new storage system for parameters $$(n'=n+1,k'=k+1,d'=d+1).$$ Let $DSS_1$ consist of nodes $v_1,\dots,v_n$, and let the stored file $F$ be of maximal size $C^{\text{exact}}_{n,k,d}(\alpha,\gamma)$.

Let then $DSS_{1+}$ denote a new system consisting of the original storage system $DSS_1$ and one extra node $v_{n+1}$ storing nothing. It is clear that $DSS_{1+}$ is a storage system for parameters $$(n+1,k+1,d+1)$$ and can store the original file $F$.

Let $\{ \sigma_j | j=1,\dots, (n+1)! \}$ be the set of permutations of the set $\{ 1, \dots, n+1 \}$. Assume that $DSS_{j}^{\text{new}}$ is a storage system for $j=1,\dots,(n+1)!$ corresponding to the permutation $\sigma_j$ such that $DSS_{j}^{\text{new}}$ is  exactly the same as $DSS_{1+}$ except that the order of the nodes is changed corresponding to the permutation $\sigma_j$, \emph{i.e.}, the $i$th node in $DSS_{1+}$ is the $\sigma_j(i)$th node in $DSS_{j}^{\text{new}}$.

Using these $(n+1)!$ new systems as building blocks we construct a new system $DSS_2$ such that its $j$th node for $j=1,\dots,n+1$ stores the $j$th node from each system $DSS_{i}^{\text{new}}$ for $i=1,\dots,(n+1)!$\,.

It is clear that this new system $DSS_2$ works for parameters $(n+1,k+1,d+1)$, has exact repair property, and stores a file of size $(n+1)! C^{\text{exact}}_{n,k,d}(\alpha,\gamma)$. By noticing that there are $n!$ such permutated copies $DSS_{j}^{\text{new}}$, where the $i$th node is empty, we get that the node size of the new system $DSS_2$ is
$$
\alpha_2=((n+1)!-n!)\alpha=n\cdot n!\alpha.
$$
Similarly, since an empty node does not need any repair we also find that the total repair bandwidth of the new system is
$$
\gamma_2=((n+1)!-n!)\gamma=n\cdot n!\gamma\,.
$$

\begin{definition}[Symmetric repair]
By \emph{symmetric repair} we mean that in the repair process of a lost node, each helper node transmits the same amount $\beta$ of information.
\end{definition}

Let us fix some repairing scheme for subsystems. Namely, define $\beta_{ijS} \in \{ 0,\beta \}$ to be the amount of information when the $i$th node repairs the $j$th node and the other helper nodes have indices from the set $S$. Now
$$
\sum_{i \in S} \beta_{ijS} = \left\{
  \begin{array}{l l}
    \gamma & \quad \text{if $j \neq n+1$}\\
    0 & \quad \text{if $j=n+1$}
  \end{array} \right.
$$
and hence
\begin{equation}
\begin{split}
\beta_2 & = \sum_{j=1}^{n+1} \sum_{\substack{ S \subseteq [n+1]\setminus\{j\} \\ |S|=d+1}} \sum_{i \in S} \beta_{ijS} \cdot (n-d-1)!d! \\
&  = \sum_{j=1}^{n} \sum_{\substack{ S \subseteq [n+1]\setminus\{j\} \\ |S|=d+1}} \gamma \cdot (n-d-1)!d! \\
&  = \sum_{j=1}^{n} \binom{n}{d+1} \gamma \cdot (n-d-1)!d! \\
&  = n \cdot n! \cdot \frac{d}{d+1} \beta.
\end{split}
\end{equation}
This proves that our construction has symmetric repair property.

The distributed storage system $DSS_1$ that we used as a starting point in our construction is not yet explicitly fixed. We have just fixed that the used storage system is some optimal system. To make it easier to follow our construction we use the notation $P_{n+1,k+1,d+1}^{\text{1, in progress}}\left(\alpha,\gamma\right)$ to denote the performance of our incomplete construction. The above reasoning implies the equality
$$
P_{n+1,k+1,d+1}^{\text{1, in progress}}\left(n\cdot n!\alpha,n\cdot n!\gamma\right) = (n+1)! C^{\text{exact}}_{n,k,d}(\alpha,\gamma).
$$
Dividing both sides by $n\cdot n!$ gives
\begin{equation}\label{firstbound}
P_{n+1,k+1,d+1}^{\text{1, in progress}}\left(\alpha,\gamma\right) = \frac{n+1}{n} C^{\text{exact}}_{n,k,d}(\alpha,\gamma).
\end{equation}

\begin{exam}\label{helppoesim}
If we relax on the requirement of a DSS to have symmetric repair then the construction becomes a bit simpler. Now, require instead only that the total repair bandwidth $\gamma$ is constant \emph{i.e.,} $\beta$ may take different values depending on the node. Let $(n,k,d)=(3,2,2)$ and $DSS_1$ be a distributed storage system with exact repair. Let $DSS_{j}^{\text{new}}$ be a storage system with $4$ nodes for $j=1,\dots,4$ where the $j$th node stores nothing, the $i$th node for $i<j$ stores as the $i$th node in the original system $DSS_1$, and the $i$th node for $i>j$ stores as the $(i-1)$th node in the original system $DSS_1$. That is, in the $j$th subsystem $DSS_{j}^{\text{new}}$ the $j$th node stores nothing while the other nodes are as those in the original system $DSS_1$.

Using these four new systems as building blocks we construct a new system $DSS_2$ for parameters $(4,3,3)$ such that its $j$th node for $j=1,\dots,4$ stores the $j$th node from each system $DSS_{i}^{\text{new}}$ for $i=1,\dots,4$. Hence each node in $DSS_2$ stores $(4-1)\alpha=3\alpha$ and the total repair bandwidth is $(4-1)\gamma=3\gamma$.

For example, if the original system $DSS_1$ consists of nodes $v_1$ storing $x$, $v_2$ storing $y$, and $v_3$ storing $x+y$ then $DSS_{1}^{\text{new}}$ consists of nodes $u_{11}$ storing nothing, $u_{12}$ storing $x_1$, $u_{13}$ storing $y_1$, and $u_{14}$ storing $x_1+y_1$. Similarly $DSS_{2}^{\text{new}}$ consists of nodes $u_{21}$ storing $x_2$, $u_{22}$ storing nothing, $u_{23}$ storing $y_2$, and $u_{24}$ storing $x_2+y_2$ and so on. Then in the resulting system the first node $w_1$ consists of nodes $u_{11}$ (storing nothing), $u_{21}$ (storing $x_2$), $u_{31}$ (storing $x_3$), and $u_{41}$ (storing $x_4$). The stored file is $(x_1,x_2,x_3,x_4,y_1,y_2,y_3,y_4)$.
\end{exam}

\begin{figure}[ht]
     \includegraphics[width=9cm]{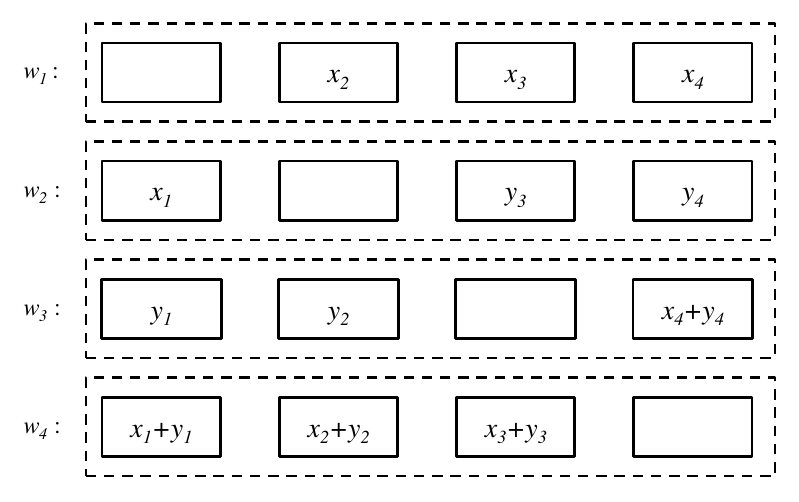}
     \caption{The figure illustrates the DSS built in Example \ref{helppoesim}. It consists of nodes $w_1$, $w_2$, $w_3$, and $w_4$.}

\end{figure}

\section{Bounds from the Construction}\label{inequalities}
Next we will derive some inequalities for the capacity in the case of exact repair. Using Equation \ref{firstbound} inductively we get
\begin{thm}\label{bound}
For an integer $j \in [ 0,k-1 ]$ we have
$$
C^{\text{exact}}_{n,k,d}\left(\alpha,\gamma\right) \geq \frac{n}{n-j} C^{\text{exact}}_{n-j,k-j,d-j}(\alpha,\gamma).
$$
\end{thm}

It is proved in \cite{kumar}, \cite{MSRequal}, \cite{optimalMDS}, and \cite{yhteensaMDS} that the MSR point can be achieved if exact repair is assumed. As a consequence of this and Theorem \ref{firstbound} we get the following bound.
\begin{thm}\label{boundMSR}
For integers $1 \leq i \leq k$ we have
$$
C^{\text{exact}}_{n,k,d}\left(\alpha,\frac{(d-k+i)\alpha}{d-k+1}\right) \geq \frac{ni\alpha}{n-k+i}\,.
$$
\end{thm}
\begin{proof}
Write $n'=n-j,k'=k-j,d'=d-j$, $\alpha=\frac{B}{k'}$, and $\gamma=\frac{d' B}{k' (d'-k'+1)}$. It is proved for the MSR point in \cite{kumar}, \cite{MSRequal}, \cite{optimalMDS}, and \cite{yhteensaMDS} that
$$
C^{\text{exact}}_{n',k',d'}(\alpha,\gamma) = B,
$$
\emph{i.e.},
$$
C^{\text{exact}}_{n-j,k-j,d-j}\left(\alpha,\frac{(d-j) \alpha}{d-k+1}\right) = (k-j)\alpha.
$$

Hence by Theorem \ref{bound} we have
$$
C^{\text{exact}}_{n,k,d}\left(\alpha,\frac{(d-j) \alpha}{d-k+1}\right) \geq \frac{n(k-j)\alpha}{n-j}\,.
$$

Now a change of  variables by setting $i=k-j$ gives us  the  result.
\end{proof}

Our construction is now ready since we have decided to use MSR optimal codes as a starting point for our construction. So let us use the notion
\begin{equation}\label{performance1}
P_{n,k,d}^{1}\left(\alpha,\frac{(d-k+i)\alpha}{d-k+1}\right) = \frac{ni\alpha}{n-k+i}
\end{equation}
for integers $i=1,\dots,k$, to note the performance of our construction.

\begin{exam}
Tian characterized the rate region of the exact-repair regenerating codes in the case $(n,k,d)=(4,3,3)$ in \cite{region433}. In this example we will compare our construction to this.

In \cite{region433} the stored file is assumed to be of size $1$ and then the rate-region of exact-regenerating codes is characterized by following pairs of $(\alpha,\beta)$: $(\frac{1}{3},\frac{1}{3})$, $(\frac{3}{8},\frac{1}{4})$, and $(\frac{1}{2},\frac{1}{6})$. These correspond to following pairs of $(\alpha,\gamma)$: $(\frac{1}{3},1)$, $(\frac{3}{8},\frac{3}{4})$, and $(\frac{1}{2},\frac{1}{2})$, \emph{i.e.}, $C^{\text{exact}}_{4,3,3}\left(\frac{1}{3},1\right) = 1$, $C^{\text{exact}}_{4,3,3}\left(\frac{3}{8},\frac{3}{4}\right) = 1$, and $C^{\text{exact}}_{4,3,3}\left(\frac{1}{2},\frac{1}{2}\right) = 1$.

Theorem \ref{boundMSR} gives in this same special case
$$
P_{4,3,3}^{1}(\alpha,i\alpha) = \frac{4i\alpha}{1+i}\,
$$
for integers $i=1,2,3$.

Hence $P_{4,3,3}^{1}\left(\alpha,\alpha\right) = 2\alpha$, $P_{4,3,3}^{1}\left(\alpha,2\alpha\right) = \frac{8\alpha}{3}$, and $P_{4,3,3}^{1}\left(\alpha,3\alpha\right) = 3\alpha$. By substituting into these $\alpha=\frac{1}{2},\frac{3}{8},\frac{1}{3}$, respectively, we get exactly the same performances as in \cite{region433}.
\end{exam}

\section{Example: Case $n=k+1=d+1$}\label{example}
In this section we study the special case $n=k+1=d+1$ and compare the resulting capacity with exact repair to the known capacity with the assumption of functional repair,
$$
C_{n-1,n-1}(\alpha,\gamma) = \sum_{j=0}^{n-2} \min \left\{ \alpha , \frac{n-1-j}{n-1} \gamma \right\}.
$$

Our construction gives codes with performance
$$
P_{n,n-1,n-1}^{1}(\alpha,i\alpha) = \frac{ni\alpha}{1+i}
$$
for integers $i=1,\dots,k$.

Notice that now in the extreme points our performance $P_{n,n-1,n-1}^{1}$ achieves the known capacity, \emph{i.e.},
$$
C^{\text{exact}}_{n,n-1,n-1}(\alpha,\alpha) = P_{n,n-1,n-1}^{1}(\alpha,\alpha) = \frac{n\alpha}{2}
$$
for the MBR point and
$$
C^{\text{exact}}_{n,n-1,n-1}(\alpha,k\alpha) = P_{n,n-1,n-1}^{1}(\alpha,k\alpha) = (n-1)\alpha
$$
for the MSR point.

As an example we study the fraction
$$
\frac{P_{n,n-1,n-1}^{1}(\alpha,i\alpha)}{C_{n-1,n-1}(\alpha,i\alpha)} = \frac{\frac{ni\alpha}{1+i}}{\sum_{j=0}^{n-2} \min \left\{ \alpha , \frac{n-1-j}{n-1} i\alpha \right\}}
$$
for integers $i \in [1,k]$. Writing it out we see that
\begin{equation}
\begin{split}
&\frac{P_{n,n-1,n-1}^{1}(\alpha,i\alpha)}{C_{n-1,n-1}(\alpha,i\alpha)} \\
= & \frac{\frac{ni}{1+i}}{\sum_{j=0}^{T } 1 + \sum_{j=T + 1}^{n-2}  \frac{n-1-j}{n-1} i } \\
= & \frac{\frac{ni}{1+i}}{ T +1 + \frac{i}{2(n-1)} \cdot (n-T-1)(n-T-2)}\,, \\
\end{split}
\end{equation}
where $T=\lfloor (n-1)(1-\frac{1}{i}) \rfloor$.

For large values of $n$ this is approximately
$$
\frac{2i^2}{2i^2+i-1} \geq \frac{8}{9}
$$
for all $i=1,\dots,k$.

Notice that if we had chosen $n=k+2=d+2$ instead of $n=k+1=d+1$, then we would have ended up with
$$
\frac{2i^2}{2i^2+3i-2}.
$$
Similarly, if we had chosen $n=k+3=d+3$ then we would have ended up with
$$
\frac{2i^2}{2i^2+5i-3}.
$$
These both are also close to 1 when $i$ is not too small. For this reason we will study the asymptotic behavior of the capacity curve more carefully in the next section.

\begin{table*}
  [ht]
  \caption{The performance of construction of Section \ref{construction}}
  \begin{tabular}{c c c}

      \includegraphics[width=5.0cm]{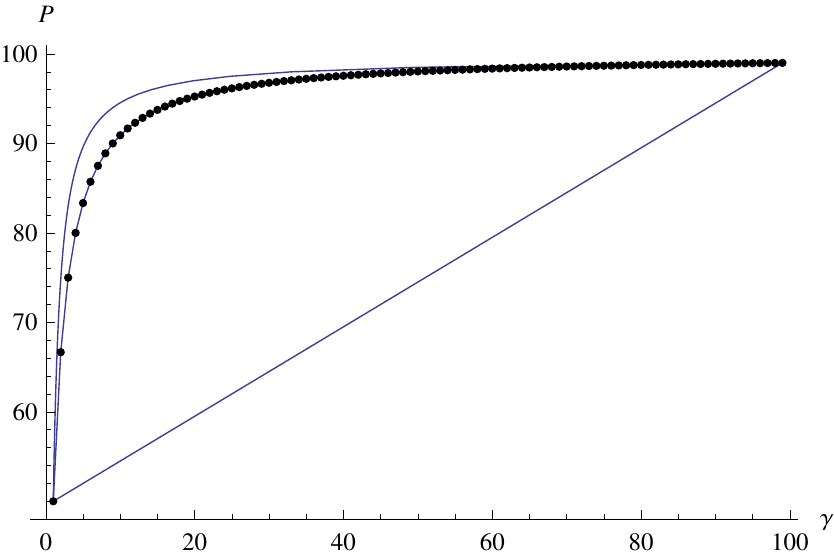} & \includegraphics[width=5.0cm]{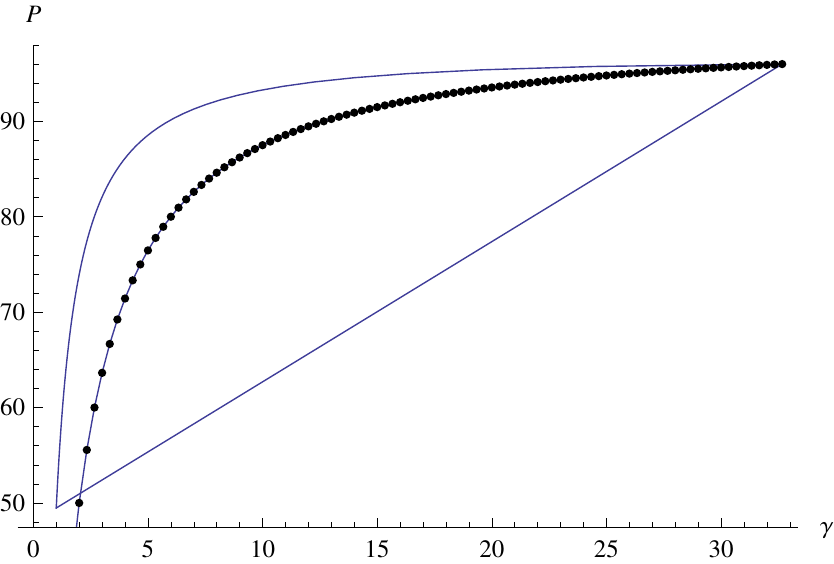} & \includegraphics[width=5.0cm]{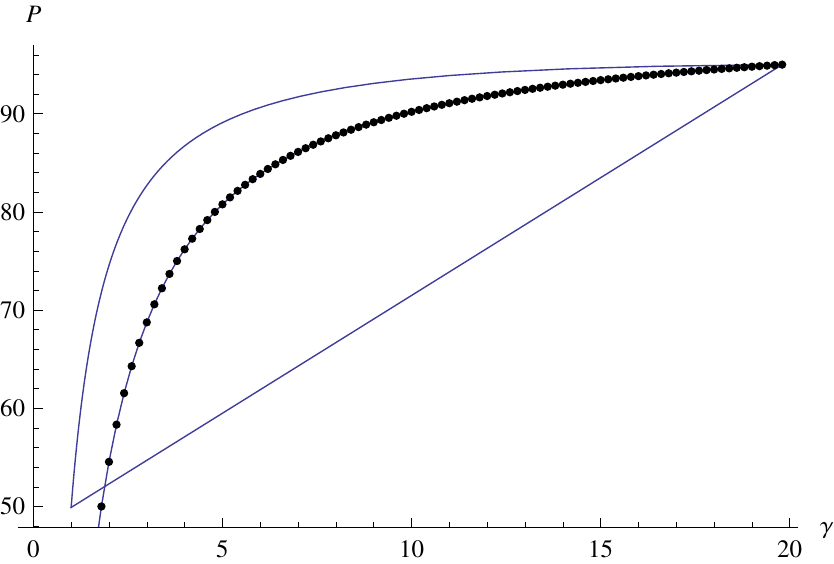} \\
      $(n,k,d)=(100,99,99)$ & $(n,k,d)=(100,96,98)$ & $(n,k,d)=(100,95,99)$ \\

      \includegraphics[width=5.0cm]{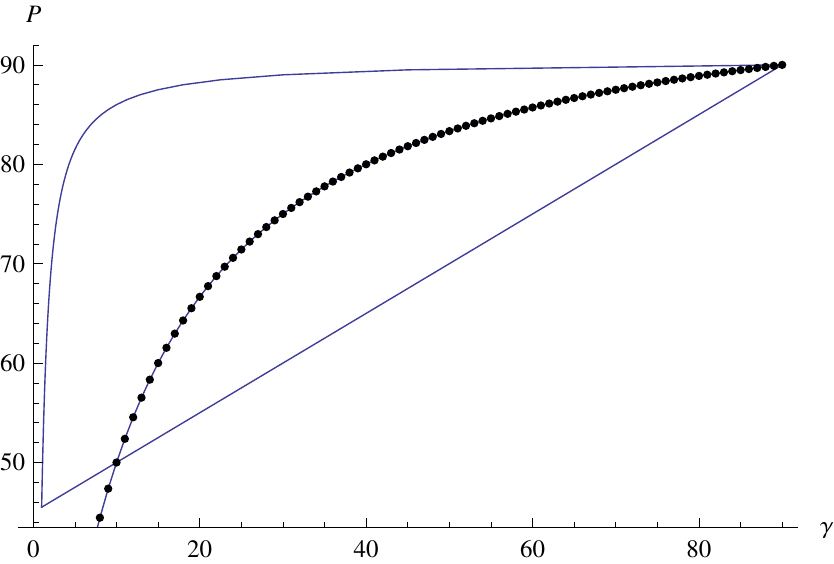} & \includegraphics[width=5.0cm]{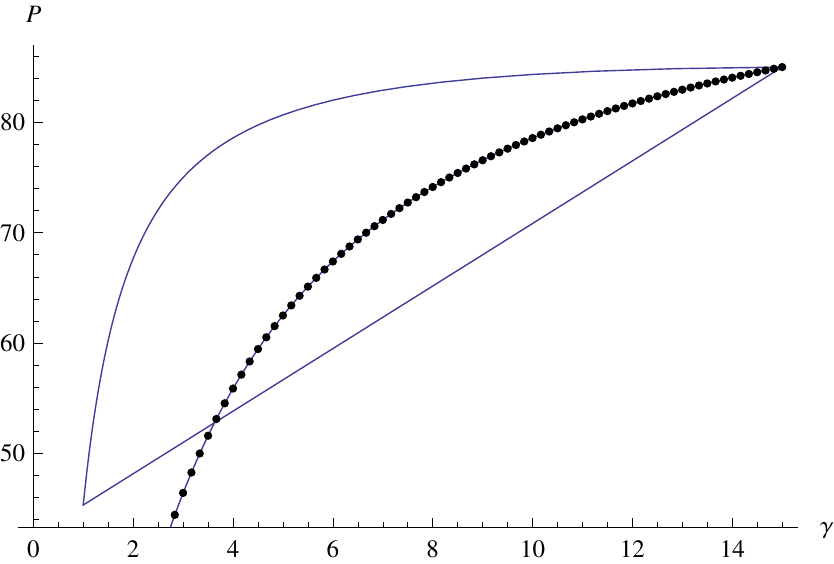} & \includegraphics[width=5.0cm]{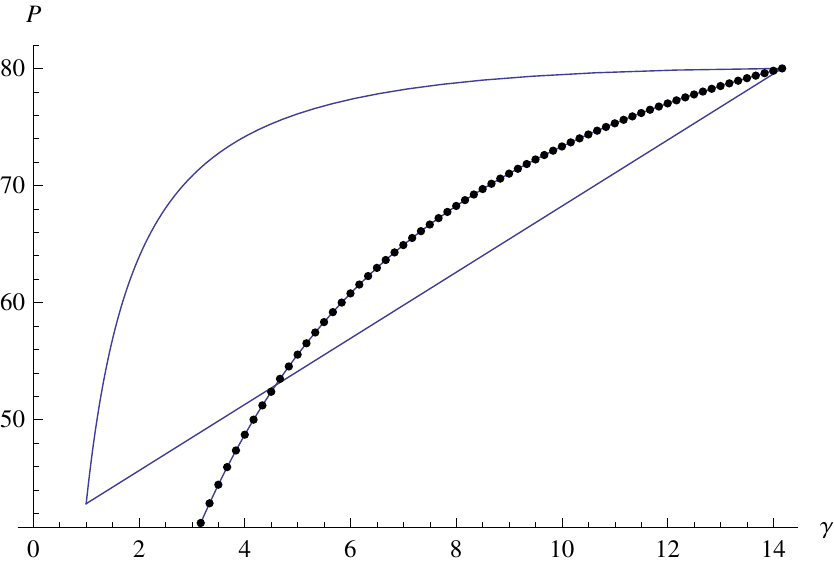} \\
      $(n,k,d)=(100,90,90)$ & $(n,k,d)=(100,85,90)$ & $(n,k,d)=(100,80,85)$ \\

  \end{tabular}
  \\[10pt]
   The figures show the performance $P_{n,k,d}^{1}$ of codes from construction of Section \ref{construction} (dotted curve) between the capacity of functionally repairing codes (uppermost curve) and the trivial lower bound given by interpolation of the known MSR and MBR points with different $(n,k,d)$. Here $\alpha=1$, and $\gamma \in [1,\frac{d}{d-k+1}]$.
   \label{tab:normalperformances}
\end{table*}

\section{The Case when $n$, $k$ and $d$ Are Close to Each Other}\label{analysis}
Next  we will study the special case where $n$, $k$ and $d$ are close to each other. We will do this by setting $n_M=n+M$, $k_M=k+M$ and $d_M=d+M$ and letting $M \rightarrow \infty$, and then examine how the capacity curve asymptotically behaves. The example in the previous section showed us that in that special case the performance $P_{n,k,d}^{1}(\alpha,\gamma)$ is quite close to the capacity of functionally regenerating codes. However, in the previous section we fixed $i$ to be an integer and then assumed that $n$ is large. In this section we tie up the values $i$ and $M$ together to arrive at a situation where the total repair bandwidth stays on a fixed point between its minimal possible value given by the MBR point and its maximal possible value given by the MSR point.

For each $M$ our construction gives a code with performance
$$
P_{n_M,k_M,d_M}^{1}\left(\alpha,\frac{(d_M-k_M+i)\alpha}{d_M-k_M+1}\right) = \frac{n_M i \alpha}{n-k+i}
$$
for $i=1,\dots,k_M$, hence $P_{n_M,k_M,d_M}^{1}\left(\alpha,\frac{(d_M-k_M+x)\alpha}{d_M-k_M+1}\right)$ with  $x \in [1,k]$ is the piecewise linear curve connecting these points.

Let $s \in (0,1]$ be a fixed number and $i=1+s(k_M-1)$. We will study how the fraction
$$
\frac{P_{n_M,k_M,d_M}^{1}\left(\alpha,\frac{(d_M-k_M+i)\alpha}{d_M-k_M+1}\right)}{C_{k_M,d_M}(\alpha,\frac{(d_M-k_M+i)\alpha}{d_M-k_M+1})}
$$
behaves as we let $M \rightarrow \infty$. Informally this tells how close our performance curve and the known capacity curve are to each other when $M$ is large, \emph{i.e.}, values $n_M,k_M,d_M$ are close to each other.

\begin{remark}
In the MSR point we have $$\gamma_{MSR}=\frac{d_M \alpha}{d_M-k_M+1}$$ and in the MBR point $$\gamma_{MBR}=\alpha.$$ Hence
$$
\alpha \cdot \frac{d_M-k_M+i}{d_M-k_M+1} = s\gamma_{MSR}+(1-s)\gamma_{MBR}.
$$
\end{remark}

\begin{thm}\label{asymptotic}
Let $s \in (0,1]$ be a fixed number and $i=1+s(k_M-1)$. Then
$$
\lim_{M \rightarrow \infty} \frac{P_{n_M,k_M,d_M}^{1}\left(\alpha,\frac{(d_M-k_M+i)\alpha}{d_M-k_M+1}\right)}{C_{k_M,d_M}(\alpha,\frac{(d_M-k_M+i)\alpha}{d_M-k_M+1})} = 1.
$$
\end{thm}

The proof is rather technical and is hence postponed to Appendix.

As a straightforward corollary to Theorem \ref{asymptotic} we have

\begin{thm}\label{asymptoticcorollary}
Let $s \in [0,1]$ be a fixed number and let $\gamma_{MSR}=\frac{d_M \alpha}{d_M-k_M+1}$ and $\gamma_{MBR}=\alpha$. Then
$$
\lim_{M \rightarrow \infty} \frac{C^{\text{exact}}_{n_M,k_M,d_M}(\alpha,s\gamma_{MSR}+(1-s)\gamma_{MBR})}{C_{k_M,d_M}(\alpha,s\gamma_{MSR}+(1-s)\gamma_{MBR})} = 1.
$$
\end{thm}

\section{A Simpler Construction}\label{easyconstructionsection}
In this section we will give a construction of a distributed storage system that again uses optimal codes at the MSR point as building blocks. There are two important differences to the main construction in Section \ref{construction} of this paper. The first difference is the easiness of the construction of this section. The second is that this construction has no symmetric repair. We only require that the total repair bandwidth is fixed to be $\gamma$ but it can consist of varying $\beta$s.

\subsection{Construction}\label{easyconstruction}
We are interested in to design a storage system for given parameters $(n,k,d)$ and $(\alpha,\gamma)$. Write
$$
\epsilon = n-k
$$
and
$$
\delta = n-d.
$$
Choose $l \in \Z_{+}$ integers $n_1, \dots, n_l$ such that
$$
n_j \geq \epsilon+1
$$
for all $j=1,\dots,l$ and $n=n_1+ \dots +n_l$. For this choice, write
$$
k_j=n_j-\epsilon
$$
and
$$
d_j=n_j-\delta
$$
for all $j=1,\dots,l$.

Assume we have $l$ storage systems $DSS_1, \dots, DSS_l$ corresponding parameters $(n_1,k_1,d_1), \dots, (n_l,k_l,d_l)$, respectively. Each of these systems has node size $\alpha$ and total repair bandwidth $\gamma$. Suppose we put these systems together to get a new bigger system $DSS_{\text{big}}$ with $n_1+ \dots +n_l=n$ nodes and storing the same $l$ files that original systems $DSS_1, \dots, DSS_l$ store.

This is indeed a distributed storage system for parameters $(n,k,d)$ and $(\alpha,\gamma)$: It is clear that we have $n$ nodes, each of size $\alpha$. Each set of $k$ nodes can recover the file: Indeed, there are $\epsilon=n-k$ nodes that are not part of the reconstruction process. Hence of each subsystem $DSS_j$ we have at least $n_j-\epsilon=k_j$ nodes that are part of the reconstruction process and hereby we can recover the corresponding file and hence the whole file.

By the same argumentation as above we notice that contacting any $d$ of the nodes we can repair a lost node. Hence we only have to download the same amount of information in the repair process of this new bigger system as in the repair process of the corresponding subsystem the total repair bandwidth is indeed $\gamma$.
\begin{remark}
The main disadvantage of constructed storage systems of this section is that they do not have symmetric repair. By shuffling the nodes corresponding to each  permutation on set $\{1,\dots,n\}$ as in the construction of Section \ref{construction} would give a DSS with symmetric repair and same performance. However, this would destroy the main advantage of this construction, namely its easiness.
\end{remark}

\begin{figure}[ht]
     \centering
     \includegraphics[width=8cm]{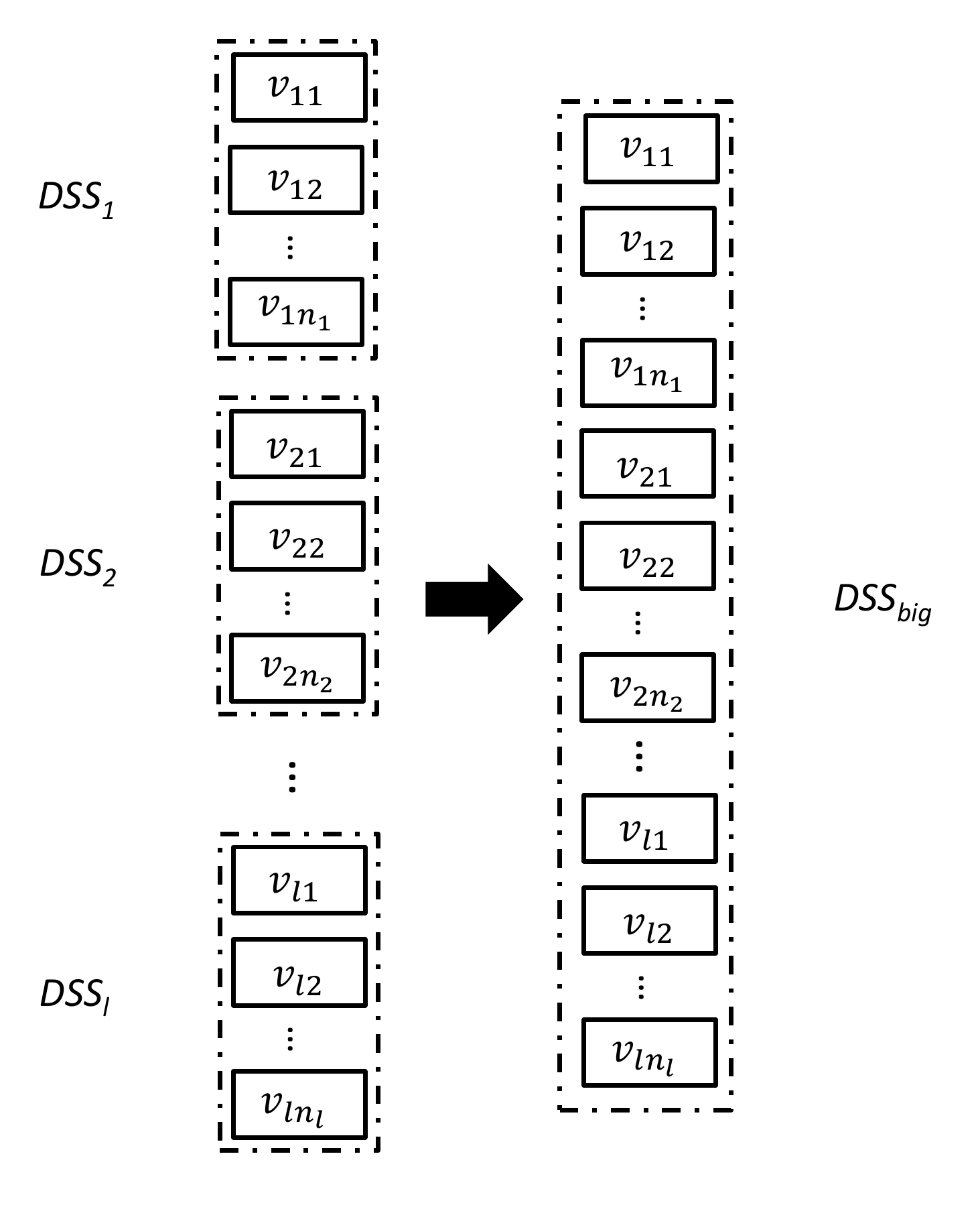}
     \caption{The figure illustrates the construction of Section \ref{easyconstructionsection}. First we have $l$ storage systems $DSS_1,\dots,DSS_l$ and then we just put them together to get a new storage system $DSS_{\text{big}}$.}

\end{figure}

\subsection{The Performance of the Construction}
In the construction we did not stick to any fixed type of subsystem. Hence we have the following general inequality.
\begin{proposition}
Given positive integers $n,k,d$ with $k \leq d < n$ and the decomposition of $n$ to positive integers $n_1,\dots,n_l$ with $n=n_1+ \dots +n_l$ and $n_j \geq n-k+1$ for all $j=1,\dots,l$. Define also integers
$$
k_j = n_j - n + k
$$
and
$$
d_j = n_j - n + d
$$
for $j=1,\dots,l$.

Then we have
\begin{equation}\label{generallowerbound}
C^{\text{exact}}_{n,k,d}(\alpha,\gamma) \geq \sum_{j=1}^{l} C^{\text{exact}}_{n_j,k_j,d_j}(\alpha,\gamma).
\end{equation}
\end{proposition}
\begin{proof}
The setup is just as in the construction of subsection \ref{easyconstruction}.
\end{proof}

To make it easier to  follow, let us use the notation $P^{\text{2, in progress}}_{n,k,d}(\alpha,\gamma)$ for the performance of this incomplete construction. By above, we have
\begin{equation}\label{eq:const2incomplete}
P^{\text{2, in progress}}_{n,k,d}(\alpha,\gamma) = \sum_{j=1}^{l} C^{\text{exact}}_{n_j,k_j,d_j}(\alpha,\gamma).
\end{equation}

Next we will fix the subsystems $DSS_1, \dots, DSS_l$ and then derive another lower bound for the performance of our construction of exact-regenerating codes. Let
$$
n_j = \left\lfloor \frac{n}{l} \right\rfloor
$$
for $j=1,\dots,l-1$ and
$$
n_l = n - (l-1) \left\lfloor \frac{n}{l} \right\rfloor.
$$
Then
$$
k_1=n_1 - n + k,
$$
$$
d_1=n_1 - n + d,
$$
$$
k_l=n_l - n + k = k - (l-1) \left\lfloor \frac{n}{l} \right\rfloor
$$
and
$$
d_l=n_l - n + d = d - (l-1) \left\lfloor \frac{n}{l} \right\rfloor.
$$
By substituting these into the equality \ref{eq:const2incomplete} we get
$$
P^{\text{2, in progress}}_{n,k,d}(\alpha,\gamma) = (l-1) C^{\text{exact}}_{n_1,k_1,d_1}(\alpha,\gamma) + C^{\text{exact}}_{n_l,k_l,d_l}(\alpha,\gamma).
$$

To finish our construction we again use MSR optimal codes as building blocks and substitute in the above
$$
\gamma = \frac{d_1 \alpha}{d_1 - k_1 + 1} = \frac{d_1 \alpha}{d - k + 1}
$$
giving $C^{\text{exact}}_{n_1,k_1,d_1}\left(\alpha,\frac{d_1 \alpha}{d - k + 1}\right)=k_1 \alpha$ and hence
$$
C^{\text{exact}}_{n,k,d}\left(\alpha,\frac{d_1 \alpha}{d - k + 1}\right) \geq (l-1) k_1 \alpha + C^{\text{exact}}_{n_l,k_l,d_l}(\alpha,\frac{d_1 \alpha}{d - k + 1}).
$$

By noticing that
$$
\alpha = \frac{(d-k+1)\gamma}{d_1} \geq \frac{(d-k+1)\gamma}{d_l}
$$
and then defining $\alpha_{\text{new}}=\frac{(d-k+1)\gamma}{d_l}=\frac{d_1 \alpha}{d_l}$, \emph{i.e.,}
$$
\gamma=\frac{d_l \alpha_{\text{new}}}{d-k+1}=\frac{d_l \alpha_{\text{new}}}{d_l-k_l+1}
$$
we find that
\begin{equation}
\begin{split}
C^{\text{exact}}_{n_l,k_l,d_l}\left(\alpha,\frac{d_1 \alpha}{d - k + 1}\right) & \geq C^{\text{exact}}_{n_l,k_l,d_l}\left(\alpha_{\text{new}},\frac{d_1 \alpha}{d - k + 1}\right) \\
& = C^{\text{exact}}_{n_l,k_l,d_l}\left(\alpha_{\text{new}},\frac{d_l \alpha_{\text{new}}}{d_l-k_l+1}\right) \\
& = k_l \alpha_{\text{new}} \\
& = \frac{k_l d_1 \alpha}{d_l}.
\end{split}
\end{equation}
Here the second to the last equality was again because of the fact that we know that the MSR point can be achieved by exact-regenerating codes.

In the calculation above giving the inequality $C^{\text{exact}}_{n_l,k_l,d_l}\left(\alpha,\frac{d_1 \alpha}{d - k + 1}\right) \geq \frac{k_l d_1 \alpha}{d_l}$ we just adapted the biggest possible MSR code when the upper bounds for node size and total repair bandwidth was given. The reason for this is that we are eager to give  a very simple construction by using already known MSR codes as building blocks.

So now we are ready to give a new lower bound for the capacity of exact-regenerating codes.
\begin{thm}
For integers $1 \leq l \leq \left\lfloor \frac{n}{n+1-k} \right\rfloor$ we have
\begin{equation}\label{toinenraja}
C^{\text{exact}}_{n,k,d}\left(\alpha,\frac{d_1 \alpha}{d - k + 1}\right) \geq \left((l-1) k_1 + \frac{k_l d_1}{d_l}\right) \alpha
\end{equation}
with
$$
k_1=\left\lfloor \frac{n}{l} \right\rfloor - n + k \text{ and } d_1=\left\lfloor \frac{n}{l} \right\rfloor - n + d
$$
and
$$
k_l=k - (l-1) \left\lfloor \frac{n}{l} \right\rfloor \text{ and } d_l=d - (l-1) \left\lfloor \frac{n}{l} \right\rfloor.
$$
\end{thm}
\begin{proof}
By the above reasoning we have the inequality \ref{toinenraja} for given $l$ if we can split our $(n,k,d)$ storage system into $l$ pieces by the above way. This is possible if we have $1 \leq k_1 \leq d_1 \leq n_1-1$ and $1 \leq k_l \leq d_l \leq n_l-1$.

The first chain of inequalities is proved by noticing that
$$
d_1=n_1 - n + d \leq n_1-1,
$$
$$
k_1=n_1 - n + k =\left\lfloor \frac{n}{l} \right\rfloor - n + k \geq \left\lfloor \frac{n}{\frac{n}{n+1-k}} \right\rfloor - n + k = 1
$$
and
$$
d_1-k_1=(n_1 - n + d)-(n_1 - n + k)=d-k \geq 0.
$$

The second chain of inequalities is proved by noticing that
$$
d_l=n_l - (n - d) \leq n_l - 1,
$$
$$
k_l \geq k - \frac{(l-1)n}{l} = k - n + \frac{n}{l} \geq k - n + \frac{n}{\frac{n}{n+1-k}} = 1
$$
and
$$
d_l - k_l = (n_l - n + d) - (n_l - n + k) = d-k \geq 0.
$$
\end{proof}

Hence the performance of our construction is
\begin{equation}\label{eq:const2}
P^{2}_{n,k,d}\left(\alpha,\frac{d_1 \alpha}{d - k + 1}\right) = \left((l-1) k_1 + \frac{k_l d_1}{d_l}\right) \alpha
\end{equation}
for $1 \leq l \leq \left\lfloor \frac{n}{n+1-k} \right\rfloor$.

\begin{exam}
Let $(n,k,d)=(3,2,2)$. Suppose a system with the first node storing $x$, second node storing $y$ and third node storing $x+y$. This MSR-optimal code storing a file $(x,y)$ has node size $\alpha=1$ and total repair bandwidth $\gamma=2\beta=2$. Take three copies of this system to form a bigger system with nine nodes:
$$x_1, y_1, x_1+y_1, x_2, y_2, x_2+y_2, x_3, y_3, x_3+y_3.$$

Similarly as in our construction this is a storage system with $(n',k',d')=(9,8,8)$, node size $\alpha'=\alpha=1$, and total repair bandwidth $\gamma=2$. It stores a file $(x_1, y_1, x_2, y_2, x_3, y_3)$ of size $6$.
\end{exam}

\begin{table*}
  [ht]
  \caption{The performance of construction of Section \ref{easyconstructionsection}}
  \begin{tabular}{c c c}

      \includegraphics[width=5.0cm]{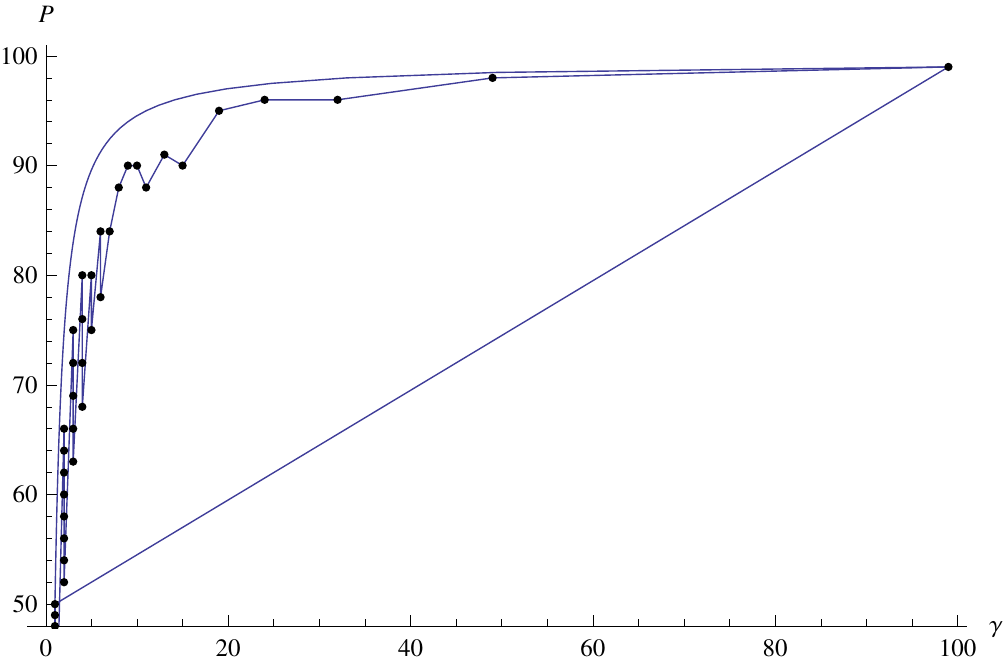} & \includegraphics[width=5.0cm]{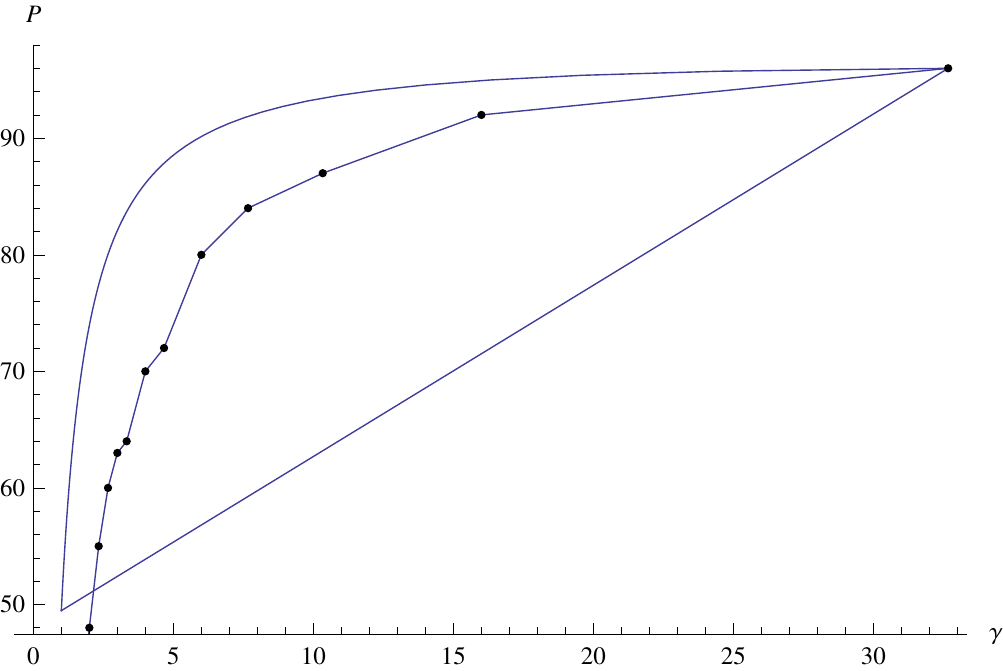} & \includegraphics[width=5.0cm]{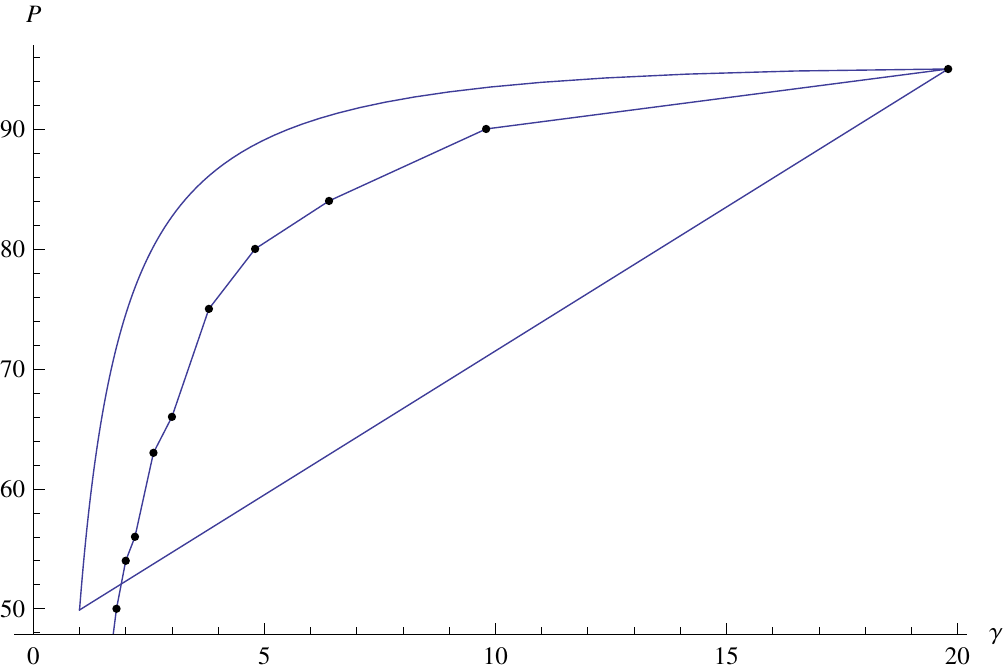} \\
      $(n,k,d)=(100,99,99)$ & $(n,k,d)=(100,96,98)$ & $(n,k,d)=(100,95,99)$ \\

      \includegraphics[width=5.0cm]{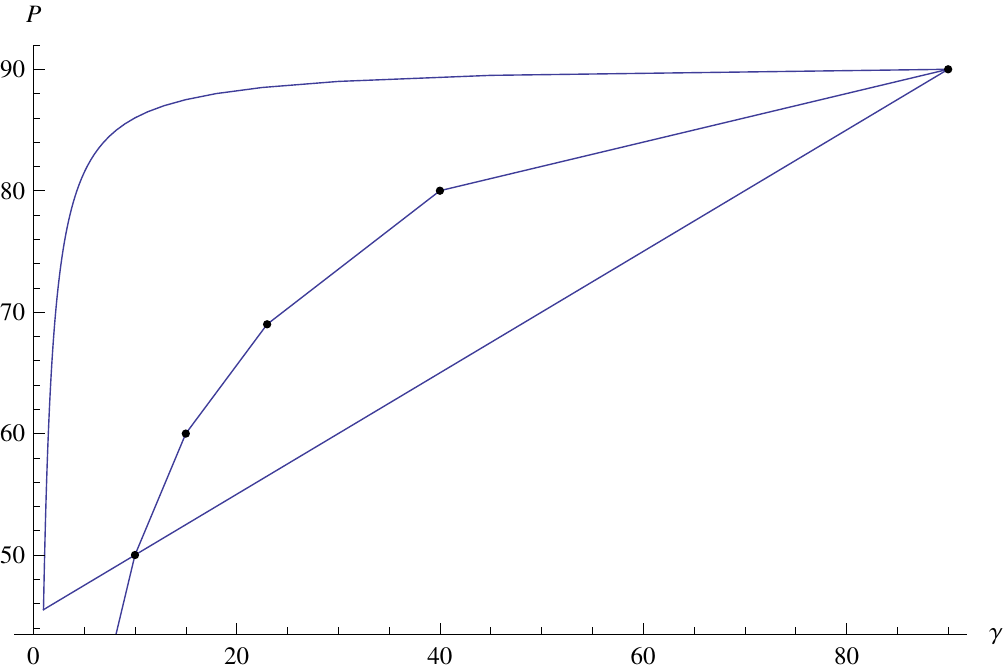} & \includegraphics[width=5.0cm]{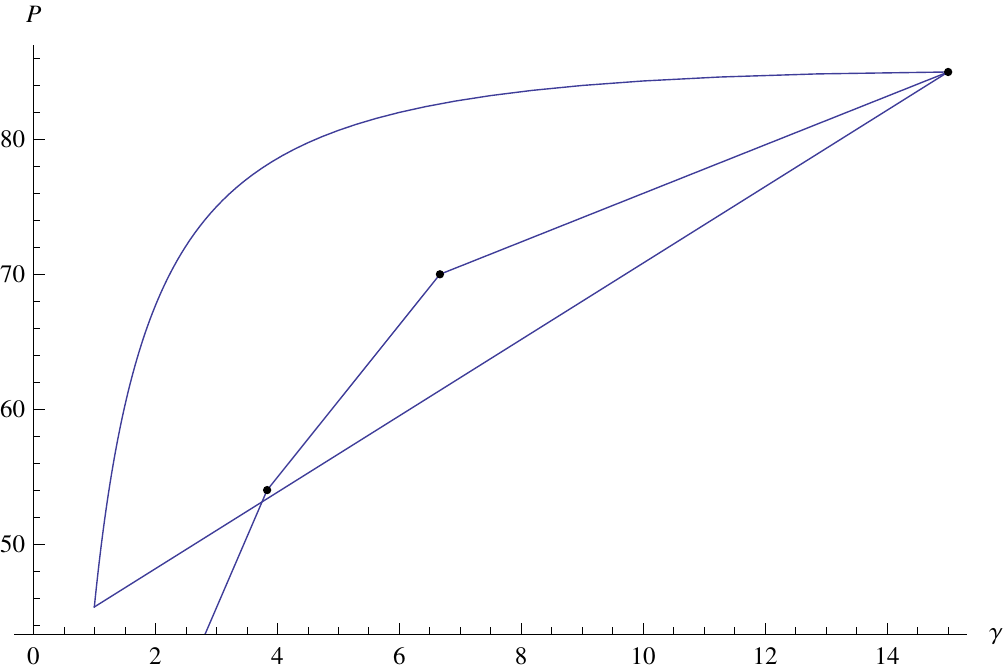} & \includegraphics[width=5.0cm]{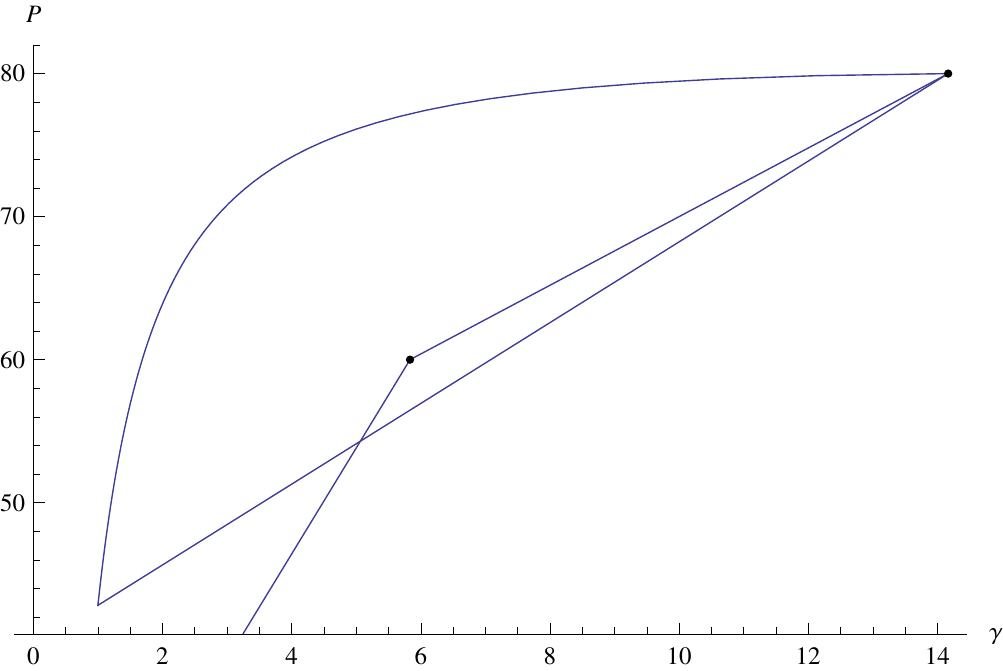} \\
      $(n,k,d)=(100,90,90)$ & $(n,k,d)=(100,85,90)$ & $(n,k,d)=(100,80,85)$ \\

  \end{tabular}
  \\[10pt]
The figures show the performance $P_{n,k,d}^{2}$ of codes from construction of Section \ref{easyconstructionsection} (dotted curve) between the capacity of functionally repairing codes (uppermost curve) and the trivial lower bound given by interpolation of the known MSR and MBR points with different $(n,k,d)$. Here $\alpha=1$, and $\gamma \in [1,\frac{d}{d-k+1}]$.
\label{tab:alternativeperformances}
\end{table*}

\subsection{Connection to the Construction of Section \ref{construction}}
Consider equality (\ref{eq:const2}) in the case $l$ divides $n$, \emph{i.e.}, $n=l n_1$. In that case we have $k_1=k_l$ and $d_1=d_l$ and hence
$$
P_{n,k,d}^{2}\left(\alpha,\frac{d_1 \alpha}{d - k + 1}\right) = l k_1 \alpha,
$$
\emph{i.e.},
\begin{equation}\label{comparison1}
P_{n,k,d}^{2}\left(\alpha,\frac{(\frac{n}{l}-n+d) \alpha}{d - k + 1}\right) = (n-nl+kl) \alpha.
\end{equation}

Let $j=k-n+\frac{n}{l}$. Since $1 \leq l \leq \left\lfloor \frac{n}{n+1-k} \right\rfloor$ we have $j \geq k-n+\frac{n}{\frac{n}{n+1-k}} = 1$ and $j \leq k-n+n = k$. Hence we can use Equation (\ref{performance1}) with this value. We get
\begin{equation}
\begin{split}
P_{n,k,d}^{1}\left(\alpha,\frac{(\frac{n}{l}-n+d) \alpha}{d - k + 1}\right) & = P_{n,k,d}^{1}\left(\alpha,\frac{(d-k+j) \alpha}{d - k + 1}\right) \\
 & = \frac{nj\alpha}{n-k+j} \\
 &  = \frac{n(k-n+\frac{n}{l})\alpha}{\frac{n}{l}} \\
 & = (n-nl+kl)\alpha \\
 & = P_{n,k,d}^{2}\left(\alpha,\frac{(\frac{n}{l}-n+d) \alpha}{d - k + 1}\right)
\end{split}
\end{equation}
so the performances $P_{n,k,d}^{1}$ and $P_{n,k,d}^{2}$ are same in this case.

This tells us that the performance of the construction of Section \ref{construction} and the performance of the construction of Section \ref{easyconstructionsection} are exactly the same whenever $l$ divides $n$, \emph{i.e.}, whenever the latter construction is built using optimal MSR codes of equal size $\frac{n}{l}$.

The explanation for the similarity of the performances of these two constructions is that the main idea of the both constructions is to increase values $k$ and $d$ but to restrain the values $\alpha$ and $\gamma$.

\section{Comparison to Similar Constructions}\label{similarconstructions}
The main idea in our construction of Section \ref{construction} was to add a new empty node in the storage system. The benefit of this was the reduction of the average node size and the average total repair bandwidth. The drawback was that we had to increase parameters $k$ and $d$. In this section we study what happens if we add something else than an empty node in the system. We try out what happens when adding an exact copy of some existing node and when adding the stored file itself.

We will see that these variations are not very useful. The performance of the construction of Subsection \ref{konstruktionodekopioinnilla} is moderate but the performance of the construction of Subsection \ref{konstruktiotiedostoilla} is not good. The key differences will be summarized in Subsection \ref{subsec:differencesummary}.

\subsection{Construction by Copying Nodes}\label{konstruktionodekopioinnilla}
Assume we have a storage system $DSS_1$ with exact repair for parameters $(n,k,d)$ with the node size $\alpha$ and the total repair bandwidth $\gamma=d\beta$. In this section we propose a construction that gives a new storage system for parameters $$(n'=n+l,k'=k+l,d'=d+l)$$ for integers $l=1,\dots,k-1$.  Let $DSS_1$ consist of the nodes $v_1,\dots,v_n$, and let the stored file $F$ be of maximal size $C^{\text{exact}}_{n,k,d}(\alpha,\gamma)$.

Let then $DSS_{1+}$ denote a new system consisting of the original storage system $DSS_1$ and $l$ extra nodes $v_{n+1},\dots,v_{n+l}$ such that $v_{n+j}$ is the exact copy of the node $v_j$ for $j=1,\dots,l$. It is clear that $DSS_{1+}$ is a storage system for parameters $$(n+l,k+l,d+l)$$ and can store the original file $F$.

Again we use permutations just similarly as in the construction of Section \ref{construction}: let $\{ \sigma_j | j=1,\dots, (n+l)! \}$ be the set of permutations of the set $\{ 1, \dots, n+l \}$. Assume that $DSS_{j}^{\text{new}}$ is a storage system for $j=1,\dots,(n+l)!$ corresponding to the permutation $\sigma_j$ such that $DSS_{j}^{\text{new}}$ is  exactly the same as $DSS_{1+}$ except that the order of the nodes is changed corresponding to the permutation $\sigma_j$, \emph{i.e.}, the $i$th node in $DSS_{1+}$ is the $\sigma_j(i)$th node in $DSS_{j}^{\text{new}}$.

Using these $(n+l)!$ new systems as building blocks we construct a new system $DSS_2$ such that its $j$th node for $j=1,\dots,n+l$ stores the $j$th node from each system $DSS_{i}^{\text{new}}$ for $i=1,\dots,(n+l)!$\,.

It is clear that this new system $DSS_2$ works for parameters $(n+l,k+l,d+l)$, has exact repair property, and stores a file of size $(n+l)! C^{\text{exact}}_{n,k,d}(\alpha,\gamma)$. The node size of the new system $DSS_2$ is
$$
\alpha_2=(n+l)! \alpha.
$$

When repairing a node there are $2l (d+l) (n+l-2)!$ subsystems in which the exact copy of the lost node is one of the helper nodes. Hence there are $(n+l)! - 2l (d+l) (n+l-2)!$ subsystems in which this not the case. So the total repair bandwidth is
$$
\gamma_2  = 2l (d+l) (n+l-2)! \alpha + ((n+l)!-2l (d+l)(n+l-2)!) \gamma
$$

Hence the performance of this incomplete construction is
$$
P^{\text{3, in progress}}_{n+l,k+l,d+l}(\alpha_2,\gamma_2) = (n+l)! C^{\text{exact}}_{n,k,d}(\alpha,\gamma)
$$
that is
\begin{equation}
P^{\text{3, in progress}}_{n+l,k+l,d+l}(\alpha,\gamma_3) = C^{\text{exact}}_{n,k,d}(\alpha,\gamma)
\end{equation}
for $l=1,\dots,k-1$ with $\gamma_3=\frac{\gamma}{(n+l)!}$, that is,
$$
\gamma_3=\frac{2l (d+l)}{(n+l)(n+l-1)} \cdot \alpha + (1-\frac{2l (d+l)}{(n+l)(n+l-1)}) \gamma.
$$

By the change of variables ($n'=n+l,k'=k+l,d'=d+l$) we obtain
\begin{equation}
P^{\text{3, in progress}}_{n,k,d}(\alpha,\gamma_4) = C^{\text{exact}}_{n-l,k-l,d-l}(\alpha,\gamma)
\end{equation}
for $l=1,\dots,\lfloor \frac{k-1}{2} \rfloor$ with
$$
\gamma_4=\frac{2l d}{n(n-1)} \cdot \alpha + (1-\frac{2l d}{n(n-1)}) \gamma.
$$

Finish again the construction by using MSR-optimal codes as a starting point. The performance we obtain is
\begin{equation}
P_{n,k,d}^{3}(\alpha,\gamma_4) = (k-l) \alpha
\end{equation}
with
$$
\gamma_4 = \left( \frac{2l d}{n(n-1)} + (1-\frac{2l d}{n(n-1)}) \frac{(d-l)}{d-k+1} \right) \alpha.
$$

\begin{table}
  [ht]
  \caption{The performance of the construction of Section \ref{konstruktionodekopioinnilla}}
  \begin{tabular}{c}

      \includegraphics[width=8cm]{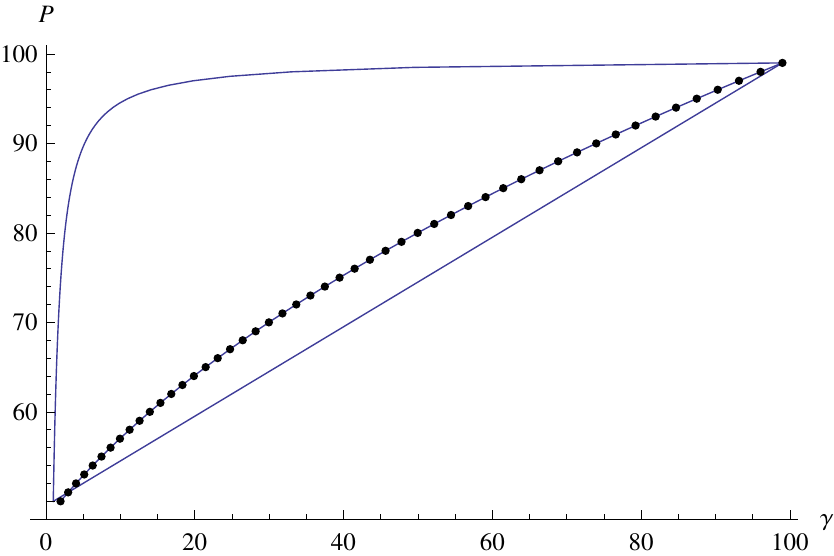}  \\
      $(n,k,d)=(100,99,99)$ \\

      \includegraphics[width=8cm]{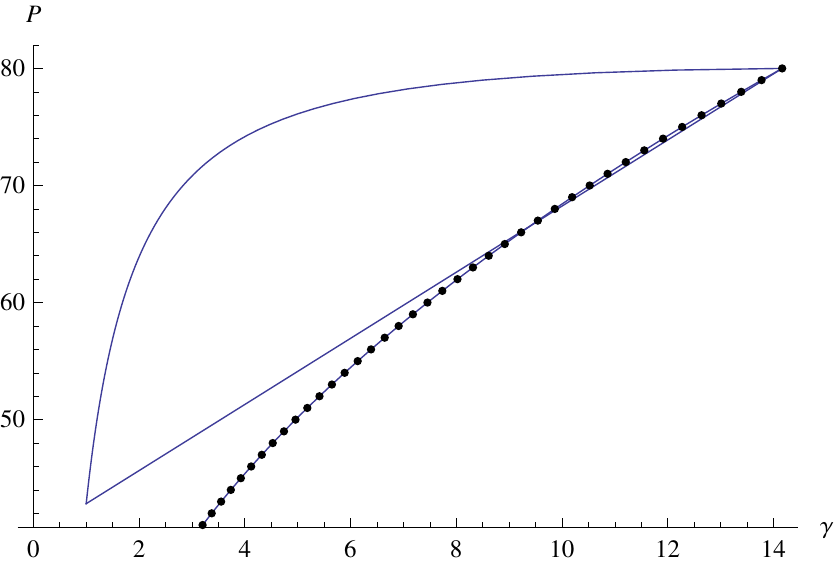}  \\
      $(n,k,d)=(100,80,85)$ \\

  \end{tabular}
  \\[10pt]
   The figure shows the performance $P_{n,k,d}^{3}$ of codes from the construction of Subsection \ref{konstruktionodekopioinnilla} (dotted curve) between the capacity of functionally repairing codes (uppermost curve) and the trivial lower bound given by interpolation of the known MSR and MBR points with different $(n,k,d)$.
   \label{tab:nodecopyperformances}
\end{table}

\subsection{Construction by Adding the File}\label{konstruktiotiedostoilla}
Assume we have a storage system $DSS_1$ with exact repair for parameters $(n,k,d)$ with the node size $\alpha$ and the total repair bandwidth $\gamma=d\beta$. In this section we propose a construction that gives a new storage system for parameters $$(n'=n+1,k'=k,d'=d).$$ Let $DSS_1$ consist of the nodes $v_1,\dots,v_n$, and let the stored file $F$ be of maximal size $C^{\text{exact}}_{n,k,d}(\alpha,\gamma)$.

Let then $DSS_{1+}$ denote a new system consisting of the original storage system $DSS_1$ and one extra node $v_{n+1}$ storing the whole file $F$. It is clear that $DSS_{1+}$ is a storage system for parameters $$(n+1,k,d)$$ and can store the original file $F$.

Again we use permutations just similarly as in the construction of Section \ref{construction}: let $\{ \sigma_j | j=1,\dots, (n+1)! \}$ be the set of permutations of the set $\{ 1, \dots, n+1 \}$. Assume that $DSS_{j}^{\text{new}}$ is a storage system for $j=1,\dots,(n+1)!$ corresponding to the permutation $\sigma_j$ such that $DSS_{j}^{\text{new}}$ is  exactly the same as $DSS_{1+}$ except that the order of the nodes is changed corresponding to the permutation $\sigma_j$, \emph{i.e.}, the $i$th node in $DSS_{1+}$ is the $\sigma_j(i)$th node in $DSS_{j}^{\text{new}}$.

Using these $(n+1)!$ new systems as building blocks we construct a new system $DSS_2$ such that its $j$th node for $j=1,\dots,n+1$ stores the $j$th node from each system $DSS_{i}^{\text{new}}$ for $i=1,\dots,(n+1)!$\,.

It is clear that this new system $DSS_2$ works for parameters $(n+1,k,d)$, has exact repair property, and stores a file of size $(n+1)! C^{\text{exact}}_{n,k,d}(\alpha,\gamma)$. By noticing that there are $n!$ such permutated copies $DSS_{j}^{\text{new}}$ where the $i$th node is storing the whole file we get that the node size of the new system $DSS_2$ is
$$
\alpha_2=((n+1)!-n!)\alpha+n!C^{\text{exact}}_{n,k,d}(\alpha,\gamma)=n! (n \alpha+C^{\text{exact}}_{n,k,d}(\alpha,\gamma))
$$
Since to repair a node storing the whole file can be done by bandwidth of size $k\alpha$ and repairing a node when the whole file is one of the helper nodes requires bandwidth $\alpha$, we find that the total repair bandwidth of the new system is
\begin{equation}
\begin{split}
\gamma_2 = & n (n-1) \cdots (n-d) \cdot (n-d)! \gamma \\
& + n d  (n-1) \cdots (n-d+1) \cdot (n-d)! \alpha \\
& + n! k\alpha \\
 = & n! ((n-d) \gamma + (d+k) \alpha )
\end{split}
\end{equation}

Hence the performance of this incomplete construction is
$$
P^{\text{4, in progress}}_{n+1,k,d}(\alpha_2,\gamma_2) = (n+1)! C^{\text{exact}}_{n,k,d}(\alpha,\gamma)
$$
that is
\begin{equation}
\begin{split}
& P^{\text{4, in progress}}_{n+1,k,d}(n \alpha+C^{\text{exact}}_{n,k,d}(\alpha,\gamma),(n-d) \gamma + (d+k)\alpha) \\
= & (n+1) C^{\text{exact}}_{n,k,d}(\alpha,\gamma).
\end{split}
\end{equation}

Substituting MSR point into above gives a code with performance
$$
P^{4}_{n+1,k,d}\left((n+k)\alpha,\left(\frac{d(n-d)}{d-k+1} + (d+k)\right)\alpha\right) =(n+1)k\alpha
$$
\emph{i.e.}
$$
P_{n+1,k,d}^{4}\left(\alpha,\frac{(nd+d-k^2+k)\alpha}{(n+k)(d-k+1)}\right) =  \frac{(n+1)k\alpha}{n+k}.
$$
However, this construction is useless because it is easy to verify that this performance is strictly worse than the trivial lower bound by timesharing when $d>k$ and it lies on the timesharing line when $k=d$.

\subsection{Summary of Differences of Different Approaches}\label{subsec:differencesummary}
Despite the clear similarities of the construction techniques, there is a huge difference on the performances $P_{n,k,d}^{1}(\alpha,\gamma)$, $P_{n,k,d}^{3}(\alpha,\gamma)$, and $P_{n,k,d}^{4}(\alpha,\gamma)$ of codes constructed using these different approaches.

In the cases where the performance $P_{n,k,d}^{1}(\alpha,\gamma)$ of the construction of Section \ref{construction} is very poor, the construction of Section \ref{konstruktiotiedostoilla} performs better. However, the performance $P_{n,k,d}^{4}(\alpha,\gamma)$ of the construction of Section \ref{konstruktiotiedostoilla} is still worse than the one achieved by timesharing of MSR and MBR points.

Comparing to the trivial lower bound given by timesharing MBR and MSR points one can summarize that the construction of Subsection \ref{konstruktiotiedostoilla} is useless, the construction of Subsection \ref{konstruktionodekopioinnilla} is in certain cases quite good, and the construction of Section \ref{construction} is in certain cases very good.

\section{Conclusions}
We have constructed exact-regenerating codes between MBR and MSR points. To the best of author's knowledge, no previous construction of exact-regenerating codes between MBR and MSR points is done except in \cite{exactrepairtian}. Compared to that construction, our construction is very different.

We have shown in this paper that when $n$, $k$, and $d$ are close to each other, the capacity of a distributed storage system when exact repair is assumed is essentially the same as when only functional repair is required. This was proved by using a specific code construction exploiting some already known codes achieving the MSR point on the trade-off curve and by studying the asymptotic behavior of the capacity curve.

A very easy alternative construction for the main construction of this paper was presented. Its performance is almost as good as the performance of the main construction and it is simple to build up. The drawback of this construction was that it has no symmetric repair.

Also we have constructed two constructions in a similar manner as the main construction. These were to be compared to the main construction. Despite the clear similarities of these three constructions their performances vary hugely.

However, when $n$, $k$, and $d$ are not close to each other then the performance of our main construction is not good when compared to the capacity of functionally repairing codes. However, there is no evidence that the capacity of a distributed storage system when exact repair is assumed is generally close to the capacity of functionally repairing codes. So as a future work it is still left to find the precise expression of the capacity of a distributed storage system when exact repair is assumed, and especially to study the behavior of the capacity when $n$, $k$, and $d$ are not close to each other.

\section{Acknowledgments} This research was partly supported by the Academy of Finland (grant \#131745) and by the Emil Aaltonen Foundation, Finland, through grants to Camilla Hollanti.

Prof. Salim El Rouayheb at the Illinois Institute of Technology is gratefully acknowledged for useful discussions. Prof. Camilla Hollanti at the Aalto University is gratefully acknowledged for useful comments on the first draft of this paper.

\section*{Appendix}\label{proofs}
\begin{proof}[The proof of Theorem \ref{asymptotic}]
Let $i=1+s(k_M-1)$. We study the behavior of the fraction for large $M$, so we have $\frac{\lfloor i \rfloor}{i} \approx 1$. Thus, to simplify the notation,  we may assume that $i$ acts as an integer. We also use the notation
$$
t=\frac{d_Ms(k_M-1)}{d-k+1+s(k_M-1)}.
$$

We have
\begin{equation}
\begin{split}
& P_{n_M,k_M,d_M}^{1}\left(\alpha,\frac{(d_M-k_M+(1+s(k_M-1)))\alpha}{d_M-k_M+1}\right) \\
= & \frac{n_M (1+s(k_M-1)) \alpha}{n-k+i}
\end{split}
\end{equation}
and
\begin{equation}
\begin{split}
& C_{k_M,d_M}\left(\alpha,\frac{(d_M-k_M+i)\alpha}{d_M-k_M+1}\right) \\
= & \alpha \left( \sum_{j=0}^{t} 1 + \sum_{j=t+1}^{k_M-1} \frac{d_M-j}{d_M} \cdot \frac{d-k+i}{d-k+1} \right) \\
= & \alpha  \left( t+1 +  \frac{(k_M-t-1)(2d+M-k-t)(d-k+i)}{2 d_M (d-k+1)} \right), \\
\end{split}
\end{equation}
whence
\begin{equation}
\begin{split}
& \frac{P_{n_M,k_M,d_M}^{1}\left(\alpha,\frac{(d_M-k_M+i)\alpha}{d_M-k_M+1}\right)}{C_{k_M,d_M}\left(\alpha,\frac{(d_M-k_M+i)\alpha}{d_M-k_M+1}\right)} \\
= & \frac{h_1(M)}{h_2(M)(h_3(M)+h_4(M))}, \\
\end{split}
\end{equation}
where
$$
h_1(M)=2n_M (1+s(k_M-1)) d_M (d-k+1),
$$
$$
h_2(M)=n-k+1+s(k_M-1),
$$
$$
h_3(M)=2(t+1)d_M(d-k+1),
$$
and
$$
h_4(M)=(k_M-t-1)(2d-k+M-t)(d-k+1+s(k_M-1)).
$$

Now it is easy to check that
$$
\frac{h_1(M)}{M^3} \rightarrow 2s(d-k+1),
$$
$$
\frac{h_2(M)}{M} \rightarrow s,
$$
and
$$
\frac{h_3(M)}{M^2} \rightarrow 2(d-k+1)
$$
as $M \rightarrow \infty.$

Note that $$M-t \approx \frac{d-k+1-ds}{s}$$ when $M$ is large and hence
\begin{equation}
\begin{split}
& \frac{h_4(M)}{M^2} \\
= & \frac{(k_M-t-1)(2d-k+M-t)}{M} \cdot \frac{d-k+1+s(k_M-1)}{M} \\
\rightarrow & 0 \cdot s=0 \\
\end{split}
\end{equation}
as $M \rightarrow \infty.$

Finally,
\begin{equation}
\begin{split}
& \frac{P_{n_M,k_M,d_M}^{1}\left(\alpha,\frac{(d_M-k_M+i)\alpha}{d_M-k_M+1}\right)}{C_{k_M,d_M}\left(\alpha,\frac{(d_M-k_M+i)\alpha}{d_M-k_M+1}\right)} \\
= & \frac{\frac{h_1(M)}{M^3}}{\frac{h_2(M)}{M} \cdot \frac{h_3(M)+h_4(M)}{M^2}} \\
\rightarrow & \frac{2s(d-k+1)}{s(2(d-k+1)+0)}=1 \\
\end{split}
\end{equation}
as $M \rightarrow \infty$, proving the claim.

\end{proof}


\end{document}